\documentclass[letterpaper, 10 pt, conference]{ieeeconf}  % Comment this line out
\usepackage{amsmath,amssymb,array,graphicx,verbatim}

\usepackage{amsthm}

\newtheorem{theorem}{Theorem}
\newtheorem{lemma}{Lemma}
\newtheorem{problem}{Problem}

\newtheorem{definition}{Definition}

\newtheorem{remark}{Remark}
\usepackage{enumerate}
\usepackage{epstopdf}
\usepackage[noadjust]{cite}
  
\usepackage{dblfloatfix}
\usepackage{dsfont}
\usepackage{breqn}

\usepackage{tikz}
\usetikzlibrary{arrows,shapes,calc,positioning}
\tikzstyle{block} = [draw,rectangle, rounded corners, minimum width=1cm, minimum height=0.8cm,text centered, line width=2pt ]
\tikzstyle{arrow} = [thick,->,>=stealth,line width=2pt]

\usepackage{circuitikz}
\usepackage{tikz-cd}
\usepackage{enumerate}
\tikzset{
  shift left/.style ={commutative diagrams/shift left={#1}},
  shift right/.style={commutative diagrams/shift right={#1}}
}
\usetikzlibrary{calc}

\usetikzlibrary{shapes}
\tikzstyle{block} = [draw,rectangle, rounded corners, minimum width=1cm, minimum height=0.8cm,text centered, line width=2pt ]
\tikzstyle{arrow} = [thick,->,>=stealth,line width=2pt]

\newcommand{\tp}{\intercal}		% transpose
\newcommand{\R}{\mathbb{R}}			% real numbers
\DeclareMathOperator{\ee}{\mathbb{E}}			% expected value
\DeclareMathOperator{\prob}{\mathbb{P}}			% probability
\DeclareMathOperator{\vecc}{\mathbf{vec}}		% probability
\DeclareMathOperator{\tr}{\mathbf{tr}}			% trace
\DeclareMathOperator{\cov}{\mathbf{cov}}		% covariance
\title{\LARGE \bf
Optimal Local and Remote Controllers with Unreliable Communication
}

\author{Yi Ouyang, Seyed Mohammad Asghari, and Ashutosh Nayyar% <-this % stops a space
\thanks{
Y. Ouyang, S. M. Asghari, and A. Nayyar are with the Department of Electrical
Engineering, University of Southern California, Los Angeles, CA 90089
(e-mail: yio@usc.edu; asgharip@usc.edu; ashutosn@usc.edu).}
 \thanks{This research was supported by NSF under grants ECCS 1509812 and CNS 1446901.}% <-this % stops a space
}

\begin{document}

\maketitle
\thispagestyle{empty}
\pagestyle{empty}

%%%%%%%%%%%%%%%%%%%%%%%%%%%%%%%%%%%%%%%%%%%%%%%%%%%%%%%%%%%%%%%%%%%%%%%%%%%%%%%%
\begin{abstract}

We consider a decentralized optimal control problem for a linear plant controlled by two controllers, a local controller and a remote controller. The local controller directly observes the state of the plant and can inform the remote controller of the plant state through a packet-drop channel. We assume that the remote controller is able to send acknowledgments to the local controller to signal the successful receipt of transmitted packets. The objective of the two controllers is to cooperatively minimize a quadratic performance cost. We provide a dynamic program for this decentralized control problem using the common information approach. Although our problem is not a partially nested LQG problem, we obtain explicit  optimal strategies for the two controllers. In the optimal strategies, both controllers compute a common estimate of the plant state based on the common information.
The remote controller's action is linear in the common estimated state, and the local controller's action is linear in both the actual state and the common estimated state.
%Even though the optimal strategies are linear, we show that the separation of control and estimation does not hold.

\end{abstract}

%%%%%%%%%%%%%%%%%%%%%%%%%%%%%%%%%%%%%%%%%%%%%%%%%%%%%%%%%%%%%%%%%%%%%%%%%%%%%%%%

\section{Introduction}

Networked control systems (NCS) are distributed systems that consist of several components (e.g. physical systems, controllers, smart sensors, etc.) and the communication network that connects them together. With the recent interest in cyber-physical systems and  the Internet of Things (IoT), NCS have received considerable attention in the recent years (see \cite{HespanhaSurvey}  and references therein). 
In contrast to traditional control systems, the interconnected components in  NCS are linked through unreliable channels with random packet drops and delays.
In the presence of unreliable communication in NCS, the implicit assumption of perfect data exchange in classical estimation and control system fails \cite{Schenato2007}. 
Therefore, efficient operation of NCS requires decentralized decision-making while taking into account the  unreliable communication among decision-makers.

In this paper, we consider an optimal control problem for a NCS consisting of a linear plant and two controllers, namely the local controller and the remote controller, connected through an unreliable communication link as shown in Fig. \ref{fig:SystemModel}.
The local controller directly observes the state of the plant and can inform the remote controller of the plant state through a channel with random packet drops. 
We consider a TCP structure so that the remote controller is able to send acknowledgments to the local controller to signal the successful receipt of transmitted packets.
The objective of the two controllers is to cooperatively minimize the overall quadratic performance cost of the NCS.
The problem is motivated from  applications that demand remote control of systems over wireless networks where links are prone to failure.
The local controller can be a small local processor proximal to the system that measures the status of the system and can perform limited control.
The remote controller can be a more powerful controller that receives information from the local processor through a wireless channel.

Similar setups of NCS has been investigated in the literature with only the remote controller present.
Various communication protocols including the TCP (where acknowledgments are available) and the UDP (where acknowledgments are not available) and variations have been investigated \cite{Imer2006optimal,Sinopoli2005,Sinopoli2006,Elia2004, Garone2008}.
For NCS with two decision-makers, \cite{LipsaMartins:2011,NayyarBasarTeneketzisVeeravalli:2013} have studied the problem when the local controller is a smart sensor and the remote controller is an estimator.
When the linear plant is controlled only by the remote controller and the local controller is a smart sensor or encoder,
\cite{BansalBasar:1989a,TatikondaSahaiMitter:2004,NairFagnaniZampieriEvan:2007,molin2013optimality,rabi2014separated} have shown that the separation of control and estimation holds for the remote controller under various communication channel models.

The problem considered in this paper is  different from previous works on NCS because 
our problem is a two-controller decentralized problem where both controllers can control the dynamics of the plant.
Finding optimal strategies for two-controller decentralized problems is generally difficult (see \cite{Witsenhausen:1968, LipsaMartins:2011b,blondel2000survey}).  In general, linear control strategies are not optimal, and even the problem of finding the best linear control strategies is not convex \cite{YukselBasar:2013}. 
Existing optimal solutions of two-controller decentralized problems require either specific information structures, such as static \cite{Radner:1962}, partially nested \cite{HoChu:1972,LamperskiDoyle:2011,LessardNayyar:2013,ShahParrilo:2013, Nayyar_Lessard_2015,Lessard_Lall_2015}, stochastically nested \cite{Yuksel:2009}, or other specific properties, such as quadratic invariance \cite{RotkowitzLall:2006} or substitutability \cite{AsghariNayyar:2015}. 
None of the above properties hold in our problem due to either the unreliable communication or the nature of dynamics and cost function.
In spite of this,  we solve the two-controller decentralized problem and provide explicit optimal strategies for the local controller and the remote controller. In the optimal strategies, both controllers compute a common estimate of the plant state based on the common information.
The remote controller's action is linear in the common estimated state, and the local controller's action is linear in both the actual state and the common estimated state.

%The main idea to the solution is to identify the common information based belief about the plant state for both controllers.
%The common belief can serve as an information state that leads to a dynamic program for optimal strategies of the two-controller problem. 

\begin{figure}[h]
\begin{center}
%\begin{tikzpicture}
%\node [rectangle,draw,minimum width=1.5cm,minimum height=1cm,line width=1pt,rounded corners]at (-1,0) (1) {{\small Controller $C^R$}}; 
%\node [rectangle,draw,minimum width=1.5cm,minimum height=1cm,line width=1pt,rounded corners]at (5,0) (2) {{\small Controller $C^L$}}; 
%\node [rectangle,draw,minimum width=7cm,minimum height=1cm,line width=1pt,rounded corners]at (2,2) (3) {Plant}; 
%
%\draw [line width=1pt] (2.west) to[cspst] (0.5,0);
%\draw [line width=1pt] (0.5,0) to (0.5,-1);
%\draw [line width=1pt] (0.5,-1) to (5,-1);
%\path[thick,->,>=stealth,line width=1pt]
%    (5,-1) edge node {}   (2.south)
%           (5,0.5) edge node {}   (5,1.5)           
%           (-1,0.5) edge node {}   (-1,1.5) 
%           (0.5,0) edge node {} (1.east)  ;
%    \path[thick,->,>=stealth, shift left=.30ex,line width=1pt]
%        (4.7,1.5) edge node {}   (4.7,0.5) ;
%\node[] at (4.4,1) {$X_t$};
%\node[] at (5.3,1) {$U_t^L$};
%\node[] at (-0.7,1) {$U_t^R$};
%\node[] at (2.3,0.3) {$\Gamma_t$};
%\end{tikzpicture}
\includegraphics[scale=1]{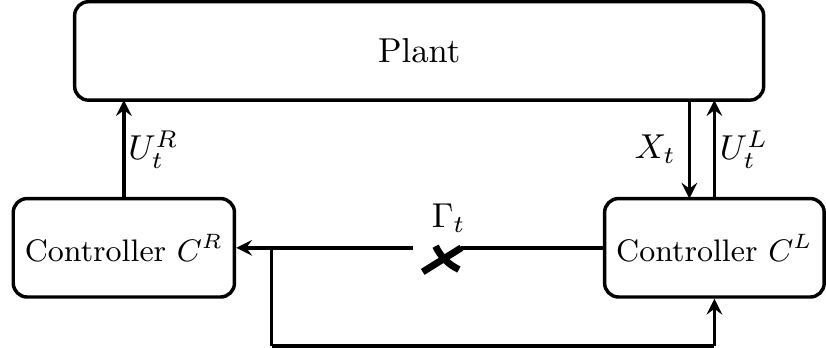}
\caption{Two-controller system model. The binary random variable $\Gamma_t$ indicates whether packets are transmitted successfully.}
\label{fig:SystemModel}
\end{center}
\end{figure}

\subsection{Organization}
The rest of the paper is organized as follows. We introduce the system model and formulate the two-controller optimal control problem in Section \ref{sec:model}. In Section \ref{sec:structure}, we provide a dynamic program for the decentralized control problem using the common information approach. 
We solve the dynamic program in Section \ref{sec:solution}.
Section \ref{sec:conclusion} concludes the paper.
% The proofs of all the technical results of the paper appear in the Appendix.

\subsection*{Notation}
Random variables/vectors are denoted by upper case letters, their realization by the corresponding lower case letter.
For a sequence of column vectors $X, Y, Z,...$, the notation $\vecc(X,Y,Z,...)$ denotes vector $[X^{\tp}, Y^{\tp}, Z^{\tp},...]^{\tp}$. The transpose and trace of matrix $A$ are denoted by $A^{\tp}$ and $\tr(A)$, respectively. 
%The notation $\textbf{0}_{n\times m}$ denotes a $n \times m$ all zeros matrix. We omit the subscripts when dimensions can be inferred from context.
In general, subscripts are used as time index while superscripts are used to index controllers.
For time indices $t_1\leq t_2$, $X_{t_1:t_2}$ (resp. $g_{t_1:t_2}(\cdot)$) is the short hand notation for the variables $(X_{t_1},X_{t_1+1},...,X_{t_2})$ (resp.  functions $(g_{t_1}(\cdot),\dots,g_{t_1}(\cdot))$).
The indicator function of set $E$ is denoted by $\mathds{1}_{E}(\cdot)$, that is, $\mathds{1}_{E}(x) = 1$ if $x \in E$, and $0$ otherwise.
$\prob(\cdot)$, $\ee[\cdot]$, and $\cov(\cdot)$ denote the probability of an event, the expectation of a random variable/vector, and the covariance matrix of a random vector, respectively.
For random variables/vectors $X$ and $Y$, $\prob(\cdot|Y=y)$ denotes the probability of an event given that $Y=y$, and $\ee[X|y] := \ee[X|Y=y]$. 
For a strategy $g$, we use $\prob^g(\cdot)$ (resp. $\ee^g[\cdot]$) to indicate that the probability (resp. expectation) depends on the choice of $g$. 
Let $\Delta(\R^n)$ denote the set of all probability measures on $\R^n$. 
%For a random vector $X \in \R^n$, $\mu(X) \in \Delta(\R^n)$ denotes the distribution of $X$ and $\cov(X)$ denotes the covariance matrix of $X$. For a vector $x \in \R^n$, $\delta(x)\in \Delta(\R^n)$ denotes the Dirac delta measure for $x$.
For any $\theta \in \Delta(\R^n)$, $\theta(E) = \int_{\R^n} \mathds{1}_{E}(x) \theta(dx)$ denotes the probability of event $E$ under $\theta$. 
The mean and the covariance of a distribution $\theta \in \Delta(\R^n)$ are denoted by $\mu(\theta)$ and $\cov(\theta)$, respectively, and are defined as $\mu(\theta) = \int_{R^n} x \theta(dx)$ and 
$\cov(\theta) = \int_{R^n} (x - \mu(\theta)) (x - \mu(\theta))^{\tp} \theta(dx)$.

\section{System Model and Problem Formulation}
\label{sec:model}
%\begin{figure}[h]
%\begin{center}
%\begin{tikzpicture}[node distance=1.5cm]
%\node (Sys) at (1,3) [block] {System};
%\node (S) at (4,3) [block] {Controller $C^L$};
%\draw [dashed, line width=2pt] (4,1.8) ellipse (0.3cm and 0.5cm);
%\node [align=left] (C) at (5.2,1.8) {Unreliable\\ Channel};
%
%\node (C) at (2,1) [block] {Controller $C^R$};
%
%%\node [align=left] (T1) at (4,1.5) {collision if more than\\ two transmit};
%
%\draw [arrow]  (Sys) --node[above]{$X_t$} (S);
%\draw [arrow]  (S) |-node[below]{$Z_t$}  (C);
%
%\draw [arrow]  ($(S.south) + (-1,0)$) -- +(0,-0.5)-|node[below]{$U^L_t$} (Sys);
%\draw [arrow]  (C.west) -- +(-1,0)|-node[left]{$U^R_t$} (Sys);
%\end{tikzpicture}
%\end{center}
%\caption{The System with Two Controllers.}
%\label{fig:system}
%\end{figure}

Consider the discrete-time system with two controllers as shown in Fig. \ref{fig:SystemModel}.
The linear plant dynamics are given by
\begin{align}
&X_{t+1} \!=\! A X_t + B^LU^L_t+ B^RU^R_t + W_t, t=0,\dots,T
 \label{Model:system}
\end{align}
where $X_t\in \R^{n_X}$ is the state of the plant at time $t$,
$U^L_t \in \R^{n_L}$ is the control action of the local controller $C^L$, $U^R_t\in \R^{n_R}$ is the control action of the remote controller $C^R$, and $A, B^L,B^R$ are matrices with appropriate dimensions.
$X_0$ is a  random vector with distribution $\pi_{X_0}$, $W_t \in \R^{n_X}$ is a zero mean  noise vector at time $t$ with distribution $\pi_{W_t}$. $X_0,W_0,W_1,\dots,W_T$ are independent random vectors with finite second moments.

 At each time $t$ the local controller $C^L$ perfectly observes the state $X_t$ and sends the observed state to the remote controller $C^R$ through an unreliable channel with packet drop probability $p$. 
 Let $\Gamma_t$ be Bernoulli random variable describing the nature of this channel, that is, $\Gamma_t=0$ when the link is broken and otherwise, $\Gamma_t=1$. We assume that $\Gamma_t$ is independent of all other variables before time $t$. Furthermore, let $Z_t$ be the channel output, then, 
 \begin{align}
\Gamma_t = &\left\{\begin{array}{ll}
1 & \text{ with probability }(1-p),\\
0 & \text{ with probability }p.
\end{array}\right.
\\
Z_t = &
\left\{\begin{array}{ll}
X_t & \text{ when } \Gamma_t = 1,\\
\emptyset & \text{ when } \Gamma_t = 0.
\end{array}\right.
\label{Model:channel}
\end{align}
We assume that the channel output $Z_t$ is perfectly observed by $C^R$. The remote controller sends an acknowledgment when it receives the state. Thus, effectively,  $Z_t$ is  perfectly observed by $C^L$ as well.  The two controllers select their control actions after observing $Z_t$. We assume that the links from the controllers to the plant are perfect.

%At each time $t$, the remote controller chooses its control action $U^R_t$ after receiving $Z_t$, and 
%the local controller chooses its control action $U^L_t$ after the reception of $Z_t$ from the ACK packet.

Let $H^L_t$ and $H^R_t$ denote the information available to $C^L$ and $C^R$ to make decisions at time $t$, respectively.\footnote{$U^R_{t-1}$ is not directly observed by $C^L$ at time $t$, but $C^L$ can obtain $U^R_{t-1}$ because $U^R_{t-1}=g^R_t(H^R_{t-1})$ and $H^R_{t-1} \subset H^L_{t}$. } Then,
\begin{align}
H^L_t= \{X_{0:t}, Z_{0:t}, U^L_{0:t-1}, U^R_{0:t-1}\}, \hspace{2mm} H^R_t = \{Z_{0:t}, U^R_{0:t-1}\}. 
\label{Model:info}
\end{align}
%Note that $Z_t$ is included in the information of local controller at time $t$ because the local controller chooses its control action after receiption of ACK packet from the remote controller. Furthermore, $Z_t$ is included in the information of remote controller at time $t$ because the remote controller chooses its control action after receiption of $Z_t$.
Let $\mathcal{H}^L_t$ and $\mathcal{H}^R_t$ be the spaces of all possible information of $C^L$ and $C^R$ at time $t$, respectively.
Then, $C^L$ and $C^R$'s actions are selected according to
\begin{align}
U^L_t = g^L_t(H^L_t), \hspace{2mm} U^R_t = g^R_t(H^R_t), 
\label{Model:strategy}
\end{align}
where the control strategies $g^L_t:\mathcal{H}^L_t\mapsto \R^{n_L}$ and $g^R_t:\mathcal{H}^R_t\mapsto \R^{n_R}$ are measurable mappings. 

%\tr\left(R_tS_tS_t^\tp\right)=

The instantaneous cost $c_t(X_t,U^L_t,U^R_t)$ of the system is a general quadratic function given by
\begin{align*}
&c_t(X_t,U^L_t,U^R_t) = 
S_t^\tp R_tS_t, ~\mbox{where}
\\
&S_t = \vecc(X_t,U^L_t,U^R_t), R_t = \left[\begin{array}{lll}
R^{XX}_t & R^{XL}_t & R^{XR}_t \\
R^{LX}_t & R^{LL}_t & R^{LR}_t \\
R^{RX}_t & R^{RL}_t & R^{RR}_t
\end{array}\right],
\end{align*}
and $R_t$ is a symmetric positive definite (PD) matrix.

The performance of strategies $g^L:=g^L_{0:T}$ and $g^R:=g^R_{0:T}$
is the total expected cost given by
\begin{align}
J(g^L,g^R)=\ee^{g^L,g^R}\left[\sum_{t=0}^T c_t(X_t,U^L_t,U^R_t)\right].
\label{Model:obj}
\end{align}

Let $\mathcal{G}^L$ and $\mathcal{G}^R$ denote all possible control strategies of $C^L$ and $C^R$ respectively. The optimal control problem for $C^L$ and $C^L$ is formally defined below.
\begin{problem}
\label{problem1}
For the system described by \eqref{Model:system}-\eqref{Model:obj}, determine control strategies $g^L$ and $g^R$ that minimize the performance cost of \eqref{Model:obj}.
%, that is,
%\begin{align}
%\inf_{g^L\in\mathcal{G}^L, g^R\in\mathcal{G}^R} J(g^L,g^R).
%\end{align}
\end{problem}

Problem \ref{problem1} is a two-controller decentralized optimal control problem.
Note that Problem \ref{problem1} is not a  partially nested LQG problem.
In particular, the local controller $C^L$'s action $U^L_{t-1}$ at $t-1$ affects $X_t$, 
and consequently, it affects $Z_t$.
%so $U^L_{t-1}$ also affects the channel output $Z_t$.
Since $Z_t$ is a part of the remote controller $C^R$'s information $H^R_t$ at $t$ but $H^L_{t-1} \not \subset H^R_t$, the information structure in Problem \ref{problem1} is not partially nested.
Therefore, linear control strategies are not necessarily optimal for Problem \ref{problem1}.

Our approach to Problem \ref{problem1} is based on the common information approach \cite{nayyar2013decentralized} for decentralized decision-making.
We identify the common belief of the system state for $C^L$ and $C^R$.
The common belief can serve as an information state that leads to a dynamic program for optimal strategies of the two-controller problem. 

\begin{remark}
The results of \cite{nayyar2013decentralized} are developed only for finite  spaces. Therefore, we can not  directly apply the results of \cite{nayyar2013decentralized} to Problem \ref{problem1}.

% since the action spaces, $\R^{n_L}$ and $\R^{n_R}$ are continuous spaces.
\end{remark}

\section{Common Belief and Dynamic Program}
\label{sec:structure}
%In order to obtain additional structure to simplify the search space of optimal strategies, we identify the common belief between $C^L$ and $C^R$ from the idea of the common information approach \cite{nayyar2013decentralized} for decentralized control.

%The common information between $C^L$ and $C^R$ at time $t$ is $H^R_t$.
%Given any realization $h^R_t \in\mathcal{H}^R_t$ of the common information, we can define the common belief $\theta_t \in \Delta(\R^{n_X})$ as the stochastic kernel of $X_t$ given $h^R_t$ and strategies before $t$. That is,
%\begin{align}
%%%&\theta_{t} = \gamma_t(h^R_t,g^L_{0:t-1},g^R_{0:t-1}),
%%%\label{eq:thetat}
%%%\\
%&\theta_t(dx_t|h^R_t) = 
%\prob^{g^L_{0:t-1},g^R_{0:t-1}}(d x_t|h^R_{t}).
%\label{eq:thetat}
%\end{align}

%In this section, we identify the common belief of the system state for the controllers $C^L$ and $C^R$. Using the common belief, we develop a dynamic program to solve the optimal control problem for the two controllers.

From \eqref{Model:info}, $H^R_t$ is the \emph{common information} among the two controllers. %controllers $C^L$ and $C^R$ have common information
%$H^R_t = H^L_t\cap H^R_t$. 
Consider fixed strategies $g^L_{0:t-1},g^R_{0:t-1}$ until time $t-1$. 
Given any realization $h^R_t \in\mathcal{H}^R_t$ of the common information, we define the common belief $\theta_t \in \Delta(\R^{n_X})$ as the conditional probability distribution of $X_t$ given $h^R_t$. That is, for any measurable set $E \subset \R^{n_X}$
\begin{align}
%%&\theta_{t} = \gamma_t(h^R_t,g^L_{0:t-1},g^R_{0:t-1}),
%%\label{eq:thetat}
%%\\
&\theta_t(X_t \in E) = 
\prob^{g^L_{0:t-1},g^R_{0:t-1}}(X_t \in E |h^R_{t}).
\label{eq:thetat}
\end{align}

Using ideas from the common information approach \cite{nayyar2013decentralized}, the common belief $\theta_t$ could serve as an information state for decentralized decision-making. We proceed to show that $\theta_t$ is indeed an information state that can be used to write a dynamic program for Problem \ref{problem1}.
%sequentially updated by the controllers' common information. 

The following Lemma provides a structural result for $C^L$.
\begin{lemma}
\label{lm:first_structure}
Let $\hat H^L_t = \vecc(X_{t}, H^R_t)$, and $\hat{\mathcal{H}}^L_t$ be  the space of all possible $\hat H^L_t$. Let $\hat{\mathcal{G}}^L=\{g^L: g^L_t \text{ is measurable from } \hat{\mathcal{H}}^L_t \text{ to }\R^{n_L}\}$.
Then,
\begin{align}
\inf_{g^L\in\mathcal{G}^L, g^R\in\mathcal{G}^R} J(g^L,g^R) =\inf_{g^L\in\hat{\mathcal{G}}^L, g^R\in\mathcal{G}^R} J(g^L,g^R).
\end{align}

\end{lemma}
From Lemma \ref{lm:first_structure}, we only need to consider strategies $g^L \in \hat{\mathcal{G}}^L$ for the local controller $C^L$. That is, $C^L$ only needs to use $\hat H^L_t = \vecc(X_{t}, H^R_t)$ to make the decision at $t$.

For any strategy $g^L \in \hat{\mathcal{G}}^L$ we provide a representation of $g^L$ using the space $\mathcal{Q}^{\theta}$ defined below.
\begin{definition}

For any $\theta \in \Delta(\R^{n_X})$, define a set of mappings
\begin{align}
&\mathcal{Q}^{\theta}
\!=\!
\left\{\!q:\R^{n_X}\!\!\mapsto\!\R^{n_L}\text{measurable,} \int_{\R^{n_X}}\hspace{-15pt}q(x) \theta (dx)\! =\! 0
\right\}.
\end{align}
%where $X^{\theta}$ is a random variable with $\mu(X^{\theta}) = \theta$.
\end{definition}

%Note that, for any strategy $g^L \in \hat{\mathcal{G}}^L$, we have
%$g^L_t(x_t,h^R_t) = \bar g^L_t(h^R_t) + \tilde g^L_t(x_t,h^R_t)$, where $\bar g^L_t(h^R_t)=\ee[g^L_t(x_t,h^R_t)|h^R_t]$ and 
%$\tilde g^L_t(x_t,h^R_t)=g^L_t(x_t,h^R_t)-\ee[g^L_t(x_t,h^R_t)|h^R_t] $.
%For each $h^R_t$, define the set of mappings $\mathcal{Q}(h^R_t)$ as 
%\begin{align}
%\mathcal{Q}(h^R_t)
%:=
%\left\{q_t:\mathcal{X}\mapsto\mathcal{U}^L\text{ measurable, } \ee[q_t(X_t)|h^R_t]=0\right\}.
%\end{align}
%Let $\bar{\mathcal{G}}^L$ and $\tilde{\mathcal{G}}^L$ be the space of strategies such that
%\begin{align}
%&\bar{\mathcal{G}}^L=\{g^L_{1:T}: \bar g^L_t \text{ is measurable from } H^R_t \text{ ro }\mathbb{R}\}
%\\
%&\tilde{\mathcal{G}}^L=\{\tilde g^L_{1:T}: \tilde g^L_t(h^R_t)\in \mathcal{Q}(h^R_t)\,\forall h^R_t \in \mathcal{H}^R_t\}
%\end{align}
\begin{lemma}
\label{lm:strategyspace}
For any strategies $g^L \in \hat{\mathcal{G}}^L$ and $g^R \in \mathcal{G}^R$, let $\theta_t$ be the conditional probability distribution defined in \eqref{eq:thetat}.
Then at any time $t$ there exists $\bar g^L_{t}: \mathcal{H}^R_t \mapsto\R^{n_L}$ and 
$\tilde g^L_{t}: \mathcal{H}^R_t\mapsto\mathcal{Q}^{\theta_t}$ such that
$\bar g^L_{t}$ is measurable and
\begin{align}
g^L_t(x_t,h^R_t) = \bar g^L_t(h^R_t) + q_t(x_t), \hspace{4mm} q_t = \tilde g^L_t(h^R_t).
\end{align}
%In other words, $\hat{\mathcal{G}}^L \subset \bar{\mathcal{G}}^L \times \tilde{\mathcal{G}}^L$.
\end{lemma}
\begin{proof}[Proof of Lemma \ref{lm:strategyspace}]
Define
\begin{align}
&\bar g^L_t(h^R_t)=\ee^{g^L_{0:t-1},g^R_{0:t-1}}
\left[g^L_t(X_t,h^R_t)|h^R_t\right],\\
&q_t(\cdot)=\tilde g^L_t(h^R_t)(\cdot)=g^L_t(\cdot,h^R_t)-\bar g^L_t(h^R_t).
\end{align}
Since $g^L_t(x_t,h^R_t)$ is measurable, $\bar g^L_t(h^R_t)$ is also measurable.  
For each $h^R_t \in \mathcal{H}^R_t$, $q_t(\cdot)=\tilde g^L_t(h^R_t)(\cdot)$ is a measurable function because $g^L_t(x_t,h^R_t)$ is measurable. Furthermore,
\begin{align*}
\int_{\R^{n_X}}q_t(x)\theta_t(d x) &=\int_{\R^{n_X}}g^L_t(x,h^R_t)\theta_t(d x) 
 \nonumber\\
&
-\ee^{g^L_{0:t-1},g^R_{0:t-1}}\left[g^L_t(X_t,h^R_t)\middle|h^R_t\right]=0.
\end{align*}
The last equality follows from \eqref{eq:thetat}. Therefore, $q_t \in \mathcal{Q}^{\theta_t}$.
\end{proof}
Note that $q_t$ belongs to $\mathcal{Q}^{\theta_t}$ and is itself a function of $h^R_t$.

From Lemma \ref{lm:strategyspace}, for any strategies $g^L \in \hat{\mathcal{G}}^L$ and $g^R \in \mathcal{G}^R$
we have a corresponding representation of the strategy $g^L_t$ of $C^L$ in terms of $\tilde{g}^L_t$ and $\bar{g}^L_t$. 
%given by, 
%\begin{align}
%(\bar g^L_{t},\tilde g^L_{t}) = RP(g^L_{0:t},g^R_{0:t}), t=0,\dots,T.
%\label{eq:gLrepresentation}
%\end{align}

%\begin{align}
%&H_t = \{Z_{0:t}, U^R_{0:t-1}\} .
%\end{align}
%Let $\mathcal{H}_t$ be the space of all possible common information at $t$.

Using the above representation of $C^L$'s strategy, we can show that the common belief $\theta_t$ is an information state with a sequential update function.

For any $x \in \R^{n_X}$, let $\delta_x\in \Delta(\R^{n_X})$ denote the Dirac delta distribution at point $x$
, that is, for any measurable set $E \subset \R^{n_X}$, $\delta_x(E) = 1$ if $x \in E$, and otherwise $\delta_x(E) = 0$.
%\begin{align*}
%\delta_x(I) = 
%\left\{\begin{array}{ll}
%1
%& \text{ if } x \in I, \\
%0 & \text{ if } x \not \in I.
%\end{array}\right.
%\end{align*}
Define $\varphi : \R^{n_X} \mapsto \Delta(\R^{n_X})$ such that $\varphi (x) = \delta_x$ for any $x \in \R^n$.
%Then, we present the sequential update function of the common belief.

\begin{lemma}
\label{lm:beliefupdate}
For any strategies $g^L \in \hat{\mathcal{G}}^L$ and $g^R \in \mathcal{G}^R$, let $(\bar g^L_{t},\tilde g^L_{t})$ be the representation of $g^L_t$ given by 
Lemma \ref{lm:strategyspace}.
%\eqref{eq:gLrepresentation} for $t=0,1,\dots,T$.
Then the common beliefs $\{\theta_{t},t=0,1,\dots,T\}$, defined by \eqref{eq:thetat}, can be sequentially updated according to
\begin{align}
\theta_0
= &\left\{\begin{array}{ll}
\pi_{X_0}
& \text{ if }z_{0}=\emptyset, 
\\
\varphi(x_0) & \text{ if }z_{0}=x_{0}.
\end{array}\right.
\label{eq:theta0}
\\
\theta_{t+1} 
%&\gamma_{t+1}(h^R_{t+1},g^L_{0:t},g^R_{0:t}) \nonumber\\
=&\psi_t(\theta_t,u^R_t,\bar u^L_t,q_t,z_{t+1}),
\label{eq:theta_update}
\end{align}
where $u^R_t,\bar u^L_t$ and $q_t$ are functions of the common information $h^R_t$ given by 
\begin{align}
u^R_t =  g^R_t(h^R_t), \hspace{2mm} \bar u^L_t = \bar g^L_t(h^R_t), \hspace{2mm} 
q_t = \tilde g^L_t(h^R_t).
\label{eq:qt}
\end{align}
Furthermore, $\psi_t(\theta_t,u^R_t,\bar u^L_t,q_t,x_{t+1}) = \varphi(x_{t+1})$ and
$\psi_t(\theta_t,u^R_t,\bar u^L_t,q_t,\emptyset)$ is a distribution on $\R^{n_X}$ such that for any measurable set $E \subset \R^{n_X}$, 
\begin{align}
&\psi_t(\theta_t,u^R_t,\bar u^L_t,q_t, \emptyset) (E)=
\nonumber\\
& \int_{\R^{n_X}} \int_{\R^{n_X}} \mathds{1}_{E}(Ax_t + B^L(\bar u^L_t+q_t(x_t))+B^R u^R_t +w_t) 
\nonumber\\
&\hspace{2cm}
\theta_t(dx_t) \pi_{W_t}(dw_t).
\label{eq:psit}
\end{align}
%\begin{align}
%&\left\{\begin{array}{ll}
%\mu\left(X_{t+1}^{(\theta_t,u^R_t,\bar u^L_t,q_t)}\right)
%& \text{ if }z_{t+1}=\emptyset, \\
%\delta(x_{t+1}) & \text{ if }z_{t+1}=x_{t+1},
%\end{array}\right.
%\label{eq:psit}
%\end{align}
%
%\begin{align}
%\theta_{t+1} =&\psi_t(\theta_t,u^R_t,\bar u^L_t,q_t,z_{t+1})
%\nonumber\\
%=&\left\{\begin{array}{ll}
%\mu\left(X_{t+1}^{(\theta_t,u^R_t,\bar u^L_t,q_t)}\right)
%& \text{ if }z_{t+1}=\emptyset, \\
%\delta(x_{t+1}) & \text{ if }z_{t+1}=x_{t+1},
%\end{array}\right.
%\label{eq:psit}
%\end{align}
%where $X_{t+1}^{(\theta_t,u^R_t,\bar u^L_t,q_t)}$ is the random variable defined by
%\begin{align}
%&X_{t+1}^{(\theta_t,u^R_t,\bar u^L_t,q_t)} \nonumber\\
%= &A X^{\theta_t}_t+B^L(\bar u^L_t+q_t(X^{\theta_t}_t))+B^R u^R_t +W_t ,
%\label{eq:tildeXtone}
%\\
%&X_t^{\theta_t} \text{ is independent of }W_t \text{ with } \mu(X_t^{\theta_t}) = \theta_t.
%\label{eq:tildeXt}
%\end{align}
%

\end{lemma}

Using the common belief and its update function, we define a class of strategies which select actions depending on the common belief $\theta_t$ instead of the entire common information $h^R_t$.
\begin{definition}
We define the set of common belief based strategies $\mathcal{G}^{C}\subset\mathcal{G}^{L}\times\mathcal{G}^{R}$.
For any $(g^L,g^R)\in \mathcal{G}^{C}$ we have the following.
At any time $t$, for each $h^R_t$, let $\theta_t$ be the common belief constructed by \eqref{eq:theta0}-\eqref{eq:qt} in Lemma \ref{lm:beliefupdate}.
Then, there exists $g^{R,C}_t:\Delta(\R^{n_X})\mapsto\R^{n_R}$,
$\bar g^{L,C}_t:\Delta(\R^{n_X})\mapsto\R^{n_L}$ and $\tilde g^{L,C}_t:\Delta(\R^{n_X})\mapsto\mathcal{Q}^{\theta_t}$ such that
\begin{align}
&g^R_t(h^R_t) = g^{R,C}_t(\theta_t),\\
&g^L_t(h^R_t) = \bar g^{L,C}_t(\theta_t)+\tilde g^{L,C}_t(\theta_t)(x_t).
\end{align}

\end{definition}

Our main result of this section is the dynamic program provided in the theorem below. 
\begin{theorem}
\label{thm:structure}
Suppose there are value functions $\{ V_{t}:\Delta(\R^{n_X})\mapsto \R \text{ for }t=0,1,\dots,T+1\}$ such that $V_{T+1} = 0$, and
for each time $t$ and for each $\theta_t \in\Delta(\R^{n_X})$ 
\begin{align}
&V_t(\theta_t) 
= \min_{ q_t\in \mathcal{Q}^{\theta_t} }\Big\lbrace \min_{\bar u^L_t\in\R^{n_L}, u^R_t\in\R^{n_R}} \Big \lbrace
\nonumber\\
& \int_{\R^{n_X}} c_t (x_t, \bar u^L_t + q_t(x_t), u_t^R) \theta_t(dx_t)
\nonumber\\
&+(1-p)  \int_{\R^{n_X}}  V_{t+1} \big(\varphi(x_{t+1})\big)  
\psi_t(\theta_t,u^R_t,\bar u^L_t,q_t, \emptyset)(dx_{t+1})
\nonumber\\
&
+ pV_{t+1}(\psi_t(\theta_t,u^R_t,\bar u^L_t,q_t, \emptyset))
\Big \rbrace \Big \rbrace.
\label{eq:DP_V}
\end{align}
If there are strategies $(g^{L*},g^{R*})\in \mathcal{G}^C$ with 
\begin{align}
&g^{R*}_t(h^R_t) = g^{R,C*}_t(\theta_t),
\label{eq:gR_new}
\\
&g^{L*}_t(h^R_t) = \bar g^{L,C*}_t(\theta_t)+\tilde g^{L,C*}_t(\theta_t)(x_t)
\label{eq:gL_new}
\end{align}
such that for each $h^R_t$
\begin{align}
\!u^{R*}_t\!=\!g^{R,C*}_t(\theta_t), \hspace{2mm}
\bar u^{L*}_t\!=\!\bar g^{L,C*}_t(\theta_t), \hspace{2mm}
q_t^*\!=\! &\tilde g^{L,C*}_t(\theta_t),
\label{eq:qstart_new}
\end{align}
achieve the minimum in the definition of $V_t(\theta_t)$,
where $\theta_t$ is the common belief constructed by \eqref{eq:theta0}-\eqref{eq:qt} in Lemma \ref{lm:beliefupdate}.
Then $g^{L*},g^{R*}$ are optimal.

\end{theorem}

Theorem \ref{thm:structure} provides a dynamic program to solve the two-controller problem. However, there are two challenges in solving the dynamic program. First, it is a dynamic program on the belief space $\Delta(\R^{n_X})$ which is infinite dimensional. 
Second, each step of the dynamic program involves a functional optimization over the functional space $\mathcal{Q}^\theta$. 
Nevertheless, in the next section, we show that it is possible to find an exact solution to the dynamic program of Theorem \ref{thm:structure}, and provide optimal strategies for the controllers.

\section{Optimal Control Strategies}
\label{sec:solution}

In this section, we identify the structure of the value function in the dynamic program \eqref{eq:DP_V}.
Using the structure, we explicitly solve the dynamic program and obtain the optimal strategies for Problem \ref{problem1}.

For a vector $x$ and a matrix $G$, we use
\begin{align}
QF(G,x)= x^\tp G\, x = \tr(Gxx^\tp)
\label{eq:QF}
\end{align}
to denote the quadratic form.

The main result of this section, stated in the theorem below, presents the structure of the value function and an explicit optimal solution of the dynamic program \eqref{eq:DP_V}.
\begin{theorem}
\label{thm:Sol_packetdrop}

For any $\theta_t$ and any time $t$, the value function of the dynamic program \eqref{eq:DP_V} in Theorem \ref{thm:structure} is given by  
\begin{align}
V_t(\theta_t) = 
&QF\left(P_t,   \mu (\theta_t)  \right) 
+ 
\tr\left(\tilde P_t \cov(\theta_t) \right)
 + e_{t}
,
\label{eq:Vt_PacketDrop}
%\end{align}
%where 
%\begin{align}
\\
e_{t} = &\sum_{s = t}^T \tr(((1-p)P_{s+1}+p\tilde P_{s+1})\cov(\pi_{W_s})),
\end{align}
and the optimal solution is given by
\begin{align}
&\left[\begin{array}{l}
\bar u^{L*}_t\\
u^{R*}_t 
\end{array}\right]
=
\left[\begin{array}{l}
\bar g^{L,C*}_t(\theta_t) \\
g^{R,C*}_t(\theta_t) 
\end{array}\right]
\nonumber\\
=&-
\left[\begin{array}{ll}
 G^{LL}_t & G^{LR}_t \\
 G^{RL}_t & G^{RR}_t
\end{array}\right]^{-1}
\left[\begin{array}{l}
G^{LX}_t \\
G^{RX}_t
\end{array}\right]
\mu (\theta_t),
\label{eq:opt_ubar}
\\
&q^{*}_t(x_t ) = \tilde g^{L,C*}_t(\theta_t)(x_t) \nonumber\\
= &-\left(\tilde G^{LL}_t\right)^{-1}\tilde G^{LX}_t\left(x_t - \mu (\theta_t) \right).
\label{eq:opt_gamma}
\end{align}
%\begin{align}
%V_t(\theta) = 
%&\tr\left(P_t  \ee\left[X_t^{\theta}(X_t^{\theta})^\tp\right]  \right) \nonumber\\
%&+ \tr\left((\tilde P_t - P_t) \cov(X_t^{\theta}) \right)
% + e_{t}
%,
%\label{eq:Vt_PacketDrop}
%\end{align}
The matrices $P_t,G_t,H_t,\tilde P_t,\tilde G_t,\tilde H_t$ defined recursively below are symmetric positive semi-definite (PSD); $G_t$ and $\tilde G_t$ are symmetric positive definite (PD).
\begin{align}
P_{T+1} = &\tilde P_{T+1} = \mathbf{0}, \text{ the all zeros matrix},
\\
P_t = &G^{XX}_t \nonumber\\
&\hspace{-1.5cm}-\left[\begin{array}{ll}
G^{XL}_t & G^{XR}_t
\end{array}\right]
\left[\begin{array}{ll}
 G^{LL}_t & G^{LR}_t \\
 G^{RL}_t & G^{RR}_t
\end{array}\right]^{-1}
\left[\begin{array}{l}
G^{LX}_t \\
G^{RX}_t
\end{array}\right],
\\
G_t = &\left[\begin{array}{lll}
G^{XX}_t & G^{XL}_t & G^{XR}_t \\
G^{LX}_t & G^{LL}_t & G^{LR}_t \\
G^{RX}_t & G^{RL}_t & G^{RR}_t
\end{array}\right]
=R_t + H_t,
\\
H_t = &\left[A,B^L,B^R\right]^\tp
P_{t+1}
\left[A,B^L,B^R\right],
\end{align}

\begin{align}
\tilde P_t = &\tilde G^{XX}_t - 
\tilde G^{XL}_t (\tilde G^{LL}_t)^{-1}\tilde G^{LX}_t,
\\
\tilde G_t  = &\left[\begin{array}{lll}
\tilde G^{XX}_t & \tilde G^{XL}_t & \tilde G^{XR}_t \\
\tilde G^{LX}_t & \tilde G^{LL}_t & \tilde G^{LR}_t \\
\tilde G^{RX}_t & \tilde G^{RL}_t & \tilde G^{RR}_t
\end{array}\right]
\nonumber\\
= &
R_t + (1-p)H_t + p\tilde H_t,
 \\
 \tilde H_t = &\left[A,B^L,B^R\right]^\tp
\tilde P_{t+1}
\left[A,B^L,B^R\right].
\end{align}
\end{theorem}

The proof of Theorem \ref{thm:Sol_packetdrop} relies on the following lemma for quadratic optimization problems.

\begin{lemma}
\label{lm:quadratic_problems}
Let $G=\left[\begin{array}{ll}
 G^{XX} & G^{XU} \\
 G^{UX} & G^{UU}
\end{array}\right]$ be a PD matrix and $P:=G^{XX} - G^{XU}  \left(G^{UU}\right)^{-1} G^{UX}$ be the Schur complement of $G^{UU}$ of $G$. 
\begin{enumerate}[(a)]
\item For any constant vector $x \in \R^n$, 
\begin{align}
 &\min_{u \in \R^m} 
QF\left(G,\vecc(x,u)\right) 
= QF\left(P, x\right)
\label{eq:QP1}
\end{align}
with optimal solution 
\begin{align}
u^* = &-\left(G^{UU}\right)^{-1} G^{UX}x.
\label{eq:QP1_sol}
\end{align}
\item For any $\theta \in \Delta(\R^n)$, let $X^{\theta}$ be a random variable with distribution $\theta$, then
\begin{align}
&\min_{q\in\mathcal{Q}^{\theta}}
\tr\left(G\cov\left(\vecc(X^{\theta},q(X^{\theta}))\right)  \right)
\nonumber\\
=& \tr\left(P 
\cov\left(X^{\theta}\right)
\right)
\label{eq:QP2}
\end{align}
with optimal solution
\begin{align}
q^*(X^{\theta}) = &-\left(G^{UU}\right)^{-1} G^{UX}\left(X^{\theta}-\mu(\theta)\right).
\label{eq:QP2_sol}
\end{align}
\end{enumerate}

\end{lemma}

Using Lemma \ref{lm:quadratic_problems}, we present a sketch of the proof of Theorem \ref{thm:Sol_packetdrop}. The complete proof is in the Appendix.
\begin{proof}[Sketch of the proof of Theorem \ref{thm:Sol_packetdrop}]
The proof is done by induction. 
Suppose the result is true at $t+1$, then at time $t$
\begin{itemize}
\item Show that $G_{t},\tilde G_{t}$ are PD.
\item Apply the induction hypothesis for \eqref{eq:Vt_PacketDrop} and the sequential update of common belief in Lemma \ref{lm:beliefupdate} to obtain
\begin{align}
&V_t(\theta_t) 
= \min_{ q_t\in \mathcal{Q}^{\theta_t} }\left\{ \min_{\bar u^L_t\in\R^{n_L}, u^R_t\in\R^{n_R}}\left\{
\vphantom{\min_{u^R_t,\bar u^L_t}}\right.\right.
 \nonumber\\
&\left.\left.
QF\left(G_t,\ee\left[S^{\theta_t}_t\right]\right)
 +\tr\left(\tilde G_{t}\cov\left(S^{\theta_t}_t\right)  \right)
\vphantom{\min_{u^R_t,\bar u^L_t}} \right\}\right\}.
\label{eq:DP_V_main}
\end{align}
In the above equation, $S^{\theta_t}_t := \vecc(X^{\theta_t},\bar u^L_t+q_t(X^{\theta_t}),u^R_t)$ where $X^{\theta_t}$ is a random vector with distribution $\theta_t$.

\item Since $q_t \in \mathcal{Q}^{\theta_t}$, $\ee [q_t(X^{\theta_t})] =0$ and consequently,
$\ee\left[S^{\theta_t}_t\right]$ depends only on $u^R_t,\bar u^L_t$. Furthermore, $\cov\left(S^{\theta_t}_t\right)$ depends only on $q_t$. Hence, 
\eqref{eq:DP_V_main} is equivalent to solving the following optimization problems
\begin{align}
 &\min_{u^R_t,\bar u^L_t} 
QF\left(G_{t},\vecc(\ee[X^{\theta_t}],\bar u^L_t,u^R_t)\right) ,
\label{eq:DP_V_P1_main}
\\
&\min_{q_t\in\mathcal{Q}_t^{\theta_t}}
\tr\left(\tilde G_{t}\cov\left(\vecc(X^{\theta_t},q_t(X^{\theta_t}),0)\right)  \right).
\label{eq:DP_V_P2_main}
\end{align}

\item Apply Lemma \ref{lm:quadratic_problems} to problems \eqref{eq:DP_V_P1_main} and \eqref{eq:DP_V_P2_main}, then we get \eqref{eq:Vt_PacketDrop} and the optimal solution at $t$.

\end{itemize}

\end{proof}

From Theorem \ref{thm:structure} and Theorem \ref{thm:Sol_packetdrop}, we can explicitly compute the optimal strategies for Problem \ref{problem1}. The optimal strategies of controllers $C^L$ and $C^R$ are shown in the following theorem.

\begin{theorem}
\label{thm:opt_strategies}
The optimal strategies of Problem \ref{problem1} are given by
\begin{align}
\left[\begin{array}{l}
\bar U^{L*}_t\\
U^{R*}_t 
\end{array}\right]
=& -
\left[\begin{array}{ll}
 G^{LL}_t & G^{LR}_t \\
 G^{RL}_t & G^{RR}_t
\end{array}\right]^{-1}
\left[\begin{array}{l}
G^{LX}_t \\
G^{RX}_t
\end{array}\right]
\hat X_t
,
\label{eq:opt_ULbarUR}
\\
U^{L*}_t
= &\bar U^{L*}_t-\left(\tilde G^{LL}_t\right)^{-1}\tilde G^{LX}_t\left(X_t - \hat X_t\right),
\label{eq:opt_UL}
\end{align}
where $\hat X_t$ is the estimate (conditional expectation) of $X_t$ based on the common information $H^R_t$. $\hat X_t$ can be computed recursively according to
\begin{align}
\hat X_0 = &\left\{
\begin{array}{ll}
\mu(\pi_{X_0}) & \text{ if }Z_{0}= \emptyset,\\
 X_{0} & \text{ if }Z_{0} = X_{0}.
\end{array}\right.
\label{eq:estimator_0}
\\
\hat X_{t+1} 
= &\left\{
\begin{array}{ll}
 A \hat X_t + B^L\bar U^{L*}_t+B^R U^{R*}_t & \text{ if }Z_{t+1}= \emptyset,\\
 X_{t+1} & \text{ if }Z_{t+1} = X_{t+1}.
\end{array}\right.
\label{eq:estimator_t}
\end{align}
\end{theorem}
Theorem \ref{thm:opt_strategies} shows that the optimal control strategy of $C^R$ is linear in the estimated state $\hat X_t$, and the optimal control strategy of $C^L$ is linear in both the actual state $X_t$ and the estimated state $\hat X_t$.
Note that even though the local controller $C^L$ perfectly observes the system state, $C^L$ still needs to compute the estimated state $\hat X_t$ to make optimal decisions.
%\begin{remark}
%%We also note that, t
%The estimation process described by \eqref{eq:estimator_0}-\eqref{eq:estimator_t} depends on the controllers' strategies. In particular, the remote controller $C^R$ needs to know $\bar U^{L*}_t$ to compute the state estimation.
%Therefore, there is no separation of control and estimation for this two-controller problem.
%\end{remark}

%\section{Numerical Experiments}
%\label{sec:numerical}

\section{Conclusion}\label{sec:conclusion}

We considered a decentralized optimal control problem for a linear plant controlled by two controllers, a local controller and a remote controller.
The local controller directly observes the state of the plant and can inform the remote controller of the plant state through a packet-drop channel with acknowledgments.
We provided a dynamic program for this decentralized control problem using the common information approach. Although our problem is not partially nested, we obtained explicit optimal strategies for the two controllers. In the optimal strategies, both controllers compute a common estimate of the plant state based on the common information.
The remote controller's action is linear in the common estimated state, and the local controller's action is linear in both the actual state and the common estimated state.
%Even though the optimal strategies are linear, we show that the separation of control and estimation does not hold.

%%%%%%%%%%%%%%%%%%%%%%%%%%%%%%%%%%%%%%%%%%%%%%%%%%%%%%%%%%%%%%%%%%%%%%%%%%%%%%%%
%\section*{Acknowledgment}

%\renewcommand{\bibsection}{\section{References}}
\bibliographystyle{ieeetr}
\bibliography{IEEEabrv,References,packet_drop,collection}

%\newpage

\appendix
\begin{proof}[Proof of Lemma \ref{lm:first_structure}]
Consider an arbitrary but fixed strategy $g^R$ of $C^R$. Then the control problem of $C^L$ becomes a MDP with state $\hat H^L_t =\vecc(X_t,H^R_t)$.
From the theory of MDP we know that $C^L$ can use only $\hat H^L_t$ to make the decision at $t$ without loss of optimality. 
\end{proof}

%\begin{proof}[Proof of Lemma \ref{lm:strategyspace}]
%Define
%\begin{align}
%&\bar g^L_t(h^R_t)=\ee^{g^L_{0:t-1},g^R_{0:t-1}}
%\left[g^L_t(X_t,h^R_t)|h^R_t\right],\\
%&q_t(\cdot)=\tilde g^L_t(h^R_t)(\cdot)=g^L_t(\cdot,h^R_t)-\bar g^L_t(h^R_t).
%\end{align}
%Since $g^L_t(x_t,h^R_t)$ is measurable, $\bar g^L_t(h^R_t)$ is also measurable.  
%For each $h^R_t \in \mathcal{H}^R_t$, $q_t(\cdot)=\tilde g^L_t(h^R_t)(\cdot)$ is a measurable function because $g^L_t(x_t,h^R_t)$ is measurable. Furthermore,
%\begin{align*}
%\int_{\R^{n_X}}q_t(x)\theta_t(d x) &=\int_{\R^{n_X}}g^L_t(x,h^R_t)\theta_t(d x) 
% \nonumber\\
%&
%-\ee^{g^L_{0:t-1},g^R_{0:t-1}}\left[g^L_t(X_t,h^R_t)\middle|h^R_t\right]=0.
%\end{align*}
%The last equality follows from \eqref{eq:thetat}. Therefore, $q_t \in \mathcal{Q}^{\theta_t}$.
%\end{proof}

\begin{proof}[Proof of Lemma \ref{lm:beliefupdate}]
%Note that according to \eqref{eq:thetat}, 
%\begin{align}
%\theta_t (X_t \in dx_t) = \prob^{g^L_{0:t-1},g^R_{0:t-1}}(X_t \in d x_t|h^R_{t}).
%\end{align}
At time $t=0$, $h_0^R = z_0$. According to \eqref{eq:thetat}, for any $E \in \R^{n_X}$,
\begin{align*}
%\label{theta0:proof}
&\theta_0 (X_0 \in E) = \prob(X_0 \in E | z_0) =  \prob(X_0 \in E | Z_0 = z_0) =
\notag \\
& \Big \lbrace \begin{array}{ll}
 \prob(X_0 \in E \vert \Gamma_0 =0) = \prob(X_0 \in E) =\pi_{X_0} (E) & \text{if }z_{0}= \emptyset,\\
 \prob(X_0 \in E \vert X_0 = x_0) = \varphi(x_0)(E) & \text{if }z_{0} = x_0
\end{array} 
\end{align*}
which gives \eqref{eq:theta0}.
At time $t+1$, for any $E \in \R^{n_X}$,
%\begin{align}
%\label{thetat:proof}
%\theta_{t+1} (X_{t+1} \in dx_{t+1}) = \prob^{g^L_{0:t},g^R_{0:t}} (X_{t+1} \in d x_{t+1} | h^R_{t+1}).
%\end{align}
if $z_{t+1} = x_{t+1}$, then
\begin{align}
\label{thetat:state:proof}
&\prob^{g^L_{0:t},g^R_{0:t}}(X_{t+1} \in E | h^R_{t+1})
= \prob (X_{t+1} \in E | x_{t+1})
\notag \\
&= 
\prob (X_{t+1} \in E | X_{t+1} = x_{t+1}) = \varphi(x_{t+1}) (E).
\end{align}
If $z_{t+1} = \emptyset$, then
\begin{align}
\label{thetat:null:proof}
&\prob^{g^L_{0:t},g^R_{0:t}}(X_{t+1} \in E | h^R_{t+1})
%\notag \\
%&= \prob^{g^L_{0:t},g^R_{0:t}}(X_{t+1} \in E | h^R_{t}, u_t^R)
\notag \\
&
= \prob^{g^L_{0:t},g^R_{0:t}}(X_{t+1} \in E | h^R_{t}, \Gamma_{t+1} =0)
\notag \\
&= \prob^{g^L_{0:t},g^R_{0:t}}(AX_t + B^L U^L_t +  B^R U^R_t + W_t \in E | h^R_{t})
\notag \\
&= \prob(AX_t + B^L (\bar u_t^L + q_t (X_t)) +  B^R u^R_t + W_t \in E | h^R_{t})
%\notag \\
%& = \int_{\R^{n_X}} \int_{\R^{n_X}} \mathds{1}_{dx_{t+1}}(Ax_t + B^L (\bar u_t^L + q_t (x_t)) +  B^R u^R_t + %w_t)
%\notag \\
%&
%\prob(X_t \in E, W_t \in dw_t \vert h_t^R) 
\notag\\
& = \int_{\R^{n_X}} \int_{\R^{n_X}} \mathds{1}_{E}(Ax_t + B^L (\bar u_t^L + q_t (x_t)) +  B^R u^R_t + w_t)
\notag \\
&
\theta_t(dx_t) \pi_{W_t}(dw_t)
\end{align}
where the third equality follows from Lemma \ref{lm:strategyspace}.
%\begin{align*}
%&u_t^L = g^L_t(x_t,h^R_t) = \bar u_t^L+ q_t(x_t), \hspace{4mm} u_t^R = g^R_t(h^R_t) \\
%&\text{ where,} \hspace{5mm} q_t = \tilde g^L_t(h^R_t), \hspace{4mm} \bar u_t^L =  \bar g^L_t(h^R_t).
%\end{align*}
Furthermore, the last equality of \eqref{thetat:null:proof} is true because $W_t$ is independent of all previous random variables, and distribution of $X_t$ given $h_t^R$ is $\theta_t$.
%hence,
%\begin{align}
%\label{conditional_prob:proof}
%&\prob(X_t \in dx_t, W_t \in dw_t \vert h_t^R) \notag \\
%& = \prob(X_t \in dx_t \vert h_t^R) \prob(W_t \in dw_t )
% = \theta_t(dx_t) \pi_{W_t}(dw_t).
%\end{align}
Note that according to \eqref{eq:thetat}, \eqref{thetat:state:proof}, and \eqref{thetat:null:proof}, $\theta_{t+1}$ is only a function of
$\theta_t, u_t^R, \bar u_t^L, q_t, z_{t+1}$. Hence, we can write it as $\theta_{t+1} = \psi_t(\theta_t,u^R_t,\bar u^L_t,q_t,z_{t+1})$
%\begin{align}
%\theta_{t+1} = \psi_t(\theta_t,u^R_t,\bar u^L_t,q_t,z_{t+1})
%\end{align}
where $\psi_t(\theta_t,u^R_t,\bar u^L_t,q_t, x_{t+1})$ is given by \eqref{thetat:state:proof} and 
$\psi_t(\theta_t,u^R_t,\bar u^L_t,q_t,\emptyset)$ is given by \eqref{thetat:null:proof}.
\end{proof}

\begin{proof}[Proof of Theorem \ref{thm:structure}]
%Suppose the dynamic problem has a solution such that the strategies $g^{L*}$ and $g^{R*}$ given by \eqref{eq:Ut} are measurable.
Suppose the strategies $(g^{L*},g^{R*})\in\mathcal{G}^C$ satisfy \eqref{eq:gR_new}-\eqref{eq:qstart_new}.
We prove by induction that for any $g^L\in\mathcal{G}^L$, $g^R\in\mathcal{G}^R$, $V_t(\prob^{g^L_{0:t-1},g^R_{0:t-1}}(d x_t|h^R_{t})) $ is a measurable function with respect to $h^R_t$, and for any information $h^R_t \in \mathcal{H}^R_t$ we have
\begin{align}
 &\ee^{g^{L}_{0:t-1},g^{R}_{0:t-1},g^{L*}_{t:T},g^{R*}_{t:T}}\left[\sum_{s=t}^Tc_s(X_s,U^L_s,U^R_s)\middle| h^R_t\right]
 \nonumber\\
=& V_t(\prob^{g^L_{0:t-1},g^R_{0:t-1}}(d x_t|h^R_{t})) 
\label{eq:Vinduction_part1}
\\
\leq & \ee^{g^L,g^R}\left[\sum_{s=t}^Tc_s(X_s,U^L_s,U^R_s)\middle| h^R_t\right].
\label{eq:Vinduction_part2}
\end{align}
Note that the above equation at $t=0$ gives the optimality of $g^{L*},g^{R*}$ for Problem \ref{problem1}.

At $T+1$ \eqref{eq:Vinduction_part1} and \eqref{eq:Vinduction_part2} are true (all terms are defined to be $0$ at $T+1$). 
Suppose \eqref{eq:Vinduction_part1} and \eqref{eq:Vinduction_part2} are true at $t+1$.

Consider any $g^L\in\mathcal{G}^L$, $g^R\in\mathcal{G}^R$ and any information $h^R_t \in \mathcal{H}^R_t$ at time $t$. Let $\theta_t(d x_t)=\prob^{g^L_{0:t-1},g^R_{0:t-1}}(d x_t|h^R_{t})$ be the common belief given $h^R_t$ under strategies $g^L_{0:t-1},g^R_{0:t-1}$. 

We first consider \eqref{eq:Vinduction_part1}.
For notational simplicity let $g'= \{g^{L}_{0:t-1},g^{R}_{0:t-1},g^{L*}_{t:T},g^{R*}_{t:T}\}$.
Let $u^{R*}_t,\bar u^{L*}_t,q_t^*$ be the minimizers defined by \eqref{eq:qstart_new} for $\theta_t$.
From the smoothing property of conditional expectation we have
\begin{align}
&\ee^{g'}\left[\sum_{s=t}^Tc_t(X_s,U^L_s,U^R_s)\middle| h^R_t\right]
\nonumber\\
= &\ee^{g'}\left[\ee^{g'}\left[\sum_{s=t+1}^Tc_s(X_s,U^L_s,U^R_s)\middle| H^R_{t+1}\right]\middle| h^R_t\right] \nonumber\\
&+\ee^{g'}\left[c_t(X_t,U^L_t,U^R_t) \middle| h^R_t\right].
\label{eq:Vinduction_part1_zero}
\end{align}
From the induction hypothesis, $V_{t+1}(\prob^{g'}(d x_{t+1}|h^R_{t+1}))$ is measurable with respect to $h^R_{t+1}$, and \eqref{eq:Vinduction_part1} holds at $t+1$. Therefore,
\begin{align}
&\ee^{g'}\left[\sum_{s=t}^Tc_t(X_s,U^L_s,U^R_s)\middle| h^R_t\right]
\nonumber\\
= &\ee^{g'}\left[V_{t+1}(\prob^{g'}(d x_{t+1}|H^R_{t+1}))\middle| h^R_t \vphantom{\ee^{g'}}\right]\nonumber\\
&+\ee^{g'}\left[c_t(X_t,U^L_t,U^R_t) \middle| h^R_t\right]
\nonumber\\
= &\ee^{g'}\left[V_{t+1}(\psi_t(\theta_t,u^{R*}_t,\bar u^{L*}_t,q^*_t,Z_{t+1})) \middle| h^R_t\right]\nonumber\\
	&+\int_{\R^{n_X}} c_t(x_t,\bar u^{L*}_t+q^*_t(x_t),u^{R*}_t)  \theta_t(dx_t).
\label{eq:Vinduction_part1_first}
\end{align}
Note that $X_{t+1}$ is independent of $\Gamma_{t+1}$.
% because it is independent of all variables before $t+1$.
Since $\prob(\Gamma_{t+1}=0) = 1-\prob(\Gamma_{t+1}=1)=p$, the first term in \eqref{eq:Vinduction_part1_first} becomes
%\begin{align}
% &\ee^{g'}\left[V_{t+1}(\psi_t(\theta_t,u^{R*}_t,\bar u^{L*}_t,q^*_t,Z_{t+1})) \middle| h^R_t\right]\nonumber\\
% = &p\ee^{g'}\left[V_{t+1}(\psi_t(\theta_t,u^{R*}_t,\bar u^{L*}_t,q^*_t,Z_{t+1})) \middle| h^R_t,\Gamma_t = 0\right]\nonumber\\
%+ &(1-p)\ee^{g'}\left[V_{t+1}(\psi_t(\theta_t,u^{R*}_t,\bar u^{L*}_t,q^*_t,Z_{t+1})) \middle| h^R_t,\Gamma_t = 1\right]\nonumber\\
%= &pV_{t+1}\left(\psi_t(\theta_t,u^{R*}_t,\bar u^{L*}_t,q^*_t,\emptyset)\right) \nonumber\\
% &+ (1-p)\ee^{g'}\left[V_{t+1}\left(\varphi(X_{t+1})\right) \middle| h^R_t\right]\nonumber\\
% = &pV_{t+1}\left(\psi_t(\theta_t,u^{R*}_t,\bar u^{L*}_t,q^*_t,\emptyset)\right) \nonumber\\
% &+ (1-p)\ee^{g'}\left[V_{t+1}\left(\varphi(X_{t+1})\right) \middle| h^R_t,\Gamma_t = 0\right]\nonumber\\
%= &pV_{t+1}\left(\psi_t(\theta_t,u^{R*}_t,\bar u^{L*}_t,q^*_t,\emptyset)\right) \nonumber\\
%&+ (1-p)\int_{\R^{n_X}} V_{t+1}(\varphi(x_{t+1})) \nonumber\\
% &\hspace{2cm}\psi_t(\theta_t,u^{R*}_t,\bar u^{L*}_t,q^*_t,\emptyset)(d x_{t+1}),
%\label{eq:Vinduction_part1_second}
%\end{align}
\begin{small}
\begin{align}
 &\ee^{g'}\left[V_{t+1}(\psi_t(\theta_t,u^{R*}_t,\bar u^{L*}_t,q^*_t,Z_{t+1})) \middle| h^R_t\right]\nonumber\\
 = &p\ee^{g'}\left[V_{t+1}(\psi_t(\theta_t,u^{R*}_t,\bar u^{L*}_t,q^*_t,Z_{t+1})) \middle| h^R_t,\Gamma_{t+1} = 0\right]\nonumber\\
+ &(1-p)\ee^{g'}\left[V_{t+1}(\psi_t(\theta_t,u^{R*}_t,\bar u^{L*}_t,q^*_t,Z_{t+1})) \middle| h^R_t,\Gamma_{t+1} = 1\right]\nonumber\\
= &pV_{t+1}\left(\alpha_t\right) 
+ (1-p)\ee^{g'}\left[V_{t+1}\left(\varphi(X_{t+1})\right) \middle| h^R_t\right]\nonumber\\
 = &pV_{t+1}\left(\alpha_t\right) 
 + (1-p)\ee^{g'}\left[V_{t+1}\left(\varphi(X_{t+1})\right) \middle| h^R_t,\Gamma_{t+1} = 0\right]\nonumber\\
= &pV_{t+1}\left(\alpha_t\right)
\!+\! (1-p)\int_{\R^{n_X}}\! V_{t+1}(\varphi(x_{t+1})) \alpha_t(d x_{t+1})
\label{eq:Vinduction_part1_second}
\end{align}
\end{small}
where $\alpha_t:=\psi_t(\theta_t,u^{R*}_t,\bar u^{L*}_t,q^*_t,\emptyset)$.
The third equality in \eqref{eq:Vinduction_part1_second} is true because $X_{t+1}$ is independent of $\Gamma_{t+1}$.
The last equality in \eqref{eq:Vinduction_part1_second} follows from Lemma \ref{lm:beliefupdate}.

Combining \eqref{eq:Vinduction_part1_first} and \eqref{eq:Vinduction_part1_second} we get \eqref{eq:Vinduction_part1} from the definition of the value function \eqref{eq:DP_V}.
Moreover, since $g'=\{g^{L}_{0:t-1},g^{R}_{0:t-1},g^{L*}_{t:T},g^{R*}_{t:T}\}$ are all measurable functions, $V_t(\prob^{g^L_{0:t-1},g^R_{0:t-1}}(d x_t|h^R_{t}))$ equals to the conditional expectation $\ee^{g'}\left[\sum_{s=t}^Tc_t(X_s,U^L_s,U^R_s)\middle| h^R_t\right]$ which is measurable with respect to $h^R_t$.

Now let's consider \eqref{eq:Vinduction_part2}.
Let $u^R_t,\bar u^L_t,q_t$ be the variables defined by \eqref{eq:qt} from $h^R_t$ and $g^{L},g^{R}$. 
Following an argument similar to that of \eqref{eq:Vinduction_part1_zero}-\eqref{eq:Vinduction_part1_second}, we get
\begin{align}
&\ee^{g^{L},g^{R}}\left[\sum_{s=t}^Tc_s(X_s,U^L_s,U^R_s)\middle| h^R_t\right]
\nonumber\\
\geq &\int_{\R^{n_X}} c_t(x_t,\bar u^{L}_t+q_t(x_t),u^{R}_t)  \theta_t(dx_t)
\nonumber\\
&+(1-p)  \int_{\R^{n_X}}  V_{t+1} \big(\varphi(x_{t+1})\big)  
\psi_t(\theta_t,u^R_t,\bar u^L_t,q_t, \emptyset)(dx_{t+1})
\nonumber\\
&
+ pV_{t+1}(\psi_t(\theta_t,u^R_t,\bar u^L_t,q_t, \emptyset))
\geq V_t(\theta_t).
\label{eq:Vinduction_part2_proof}
\end{align}
The last inequality in \eqref{eq:Vinduction_part2_proof} follows from the definition of the value function \eqref{eq:DP_V}.
This completes the proof of the induction step, and the proof of the theorem.
\end{proof}

\begin{proof}[Proof of Lemma \ref{lm:quadratic_problems}]
The proof of the first part of Lemma \ref{lm:quadratic_problems} is trivial.
%First consider the optimization problem \eqref{eq:QP1}, the first part of Lemma \ref{lm:quadratic_problems}.
%$G^{UU}$ is invertible since $G$ is PD. Therefore, the objective function in \eqref{eq:QP1} is equal to
%\begin{align}
% &QF\left(G,\vecc(x,u)\right) 
%\nonumber\\
%=&QF\left(G^{XX},x\right)+
%QF\left(G^{UU},u+\left(G^{UU}\right)^{-1} G^{UX}x\right) 
%\nonumber\\
%&-QF\left(G^{UU},\left(G^{UU}
%\right)^{-1} G^{UX}x\right) 
%\nonumber\\
%=&QF\left(P,x\right)
%+QF\left(G^{UU},u+\left(G^{UU}\right)^{-1} G^{UX}x\right).
%\end{align}
%Since $G^{UU}$ is PD, the second term in the above equation is greater than or equal to zero.
%Then it is clear that the minimum of \eqref{eq:QP1} is $QF\left(P,x\right)$, and it is achieved by the solution given by \eqref{eq:QP1_sol}.

Now let's consider the second part of Lemma \ref{lm:quadratic_problems}, the functional optimization problem \eqref{eq:QP2}. From the property of trace and covariance we have
\begin{align}
&
\tr\left(G\cov\left(\vecc\left(X^{\theta},q(X^{\theta})\right)\right)  \right)
\nonumber\\
=& \ee\left[
QF\left(
G, \vecc\left(X^{\theta},q(X^{\theta})\right)
-\ee\left[\vecc\left(X^{\theta},q(X^{\theta})\right)\right]
\right)
\right]
\nonumber\\
=& \ee\left[
QF\left(
G, \vecc\left(X^{\theta}-\mu(\theta),q(X^{\theta})\right)
\right)
\right]
\label{eq:OP2_obj1}
\end{align}
where the last equation in \eqref{eq:OP2_obj1} holds because
$\ee\left[q(X^{\theta})\right]=0$.
Since $\theta$ is the distribution of $X^{\theta})$, we have
\begin{align}
&\ee\left[
QF\left(
G, \vecc\left(X^{\theta}-\mu(\theta),q(X^{\theta})\right)
\right)
\right]
\nonumber\\
=&\int_{\R^n}
QF\left(
G, \vecc\left(y-\mu(\theta),q(y)\right)
\right)
\theta(dy)
\label{eq:OP2_obj2}
\end{align}
Note that the function inside the integral of \eqref{eq:OP2_obj2} has the quadratic form of the optimization problem \eqref{eq:QP1} with $x = y - \mu(\theta)$ and $u=q(y)$.
From the results of the first part of Lemma \ref{lm:quadratic_problems}, for any $y\in\R^n$ we have
\begin{align*}
&QF\left(
G, \vecc\left(y-\mu(\theta),q(y)\right)
\right)
\nonumber\\
\geq 
&QF\left(
G, \vecc\left(y-\mu(\theta),q^*(y)\right)
\right)
=
QF\left(P,y-\mu(\theta)\right)
%\label{eq:OP2_obj3}
\end{align*}
where $q^*$ is the function given by \eqref{eq:QP2_sol}.
%\begin{align}
%q^*(y) = -
%\left(G^{UU}\right)^{-1} G^{UX}
%\left(y-\ee\left[X^{\theta}\right]\right).
%\end{align}
It is clear that $q^*_t$ is measurable. Furthermore, 
\begin{align*}
\ee\left[q^*(X^{\theta})\right]
= &\int_{\R^{n}}
-\left(G^{UU}\right)^{-1} G^{UX}
\left(x-\mu(\theta)\right)
\theta(dx)
=0.
\end{align*}
Consequently, $q^*\in\mathcal{Q}^{\theta}$. Then $q^*$ is the optimal solution to problem \eqref{eq:QP2}, and the optimal value is given by
\begin{align}
&\int_{\R^n}
QF\left(
G, \vecc\left(y-\mu(\theta),q^*(y)\right)
\right)
\theta(dy)
\nonumber\\
=&
\int_{\R^n}
QF\left(P,y-\mu(\theta)\right)
\theta(dy)
=\ee\left[
QF\left(P,X^{\theta}-\mu(\theta)\right)
\right]
\nonumber\\
=& \tr\left(P 
\cov\left(\theta\right)
\right).
\label{eq:QP2_inproof}
\end{align}
\end{proof}

\begin{proof}[Proof of Theorem \ref{thm:Sol_packetdrop}]
The proof is done by induction. 

At $T+1$, \eqref{eq:Vt_PacketDrop} is true since $P_{T+1} = \tilde P_{T+1} = \mathbf{0}$.
Suppose \eqref{eq:Vt_PacketDrop} is true at $t+1$ and the matrices are all PSD and $G_{t+1},\tilde G_{t+1}$ are PD.

At time $t$, since $P_{t+1}$ and $\tilde P_t$ are PSD, $H_t$ and $\tilde H_t$ are PSD. 
Since $R_t$ is PD, $G_t = R_t+H_t$ and $\tilde G_t=R_t+(1-p)H_t+p\tilde H_t$ are also PD. Then $P_t$ is PSD because $P_t$ is the Schur complement of $\left[\begin{array}{ll}
 G^{LL}_t & G^{LR}_t \\
 G^{RL}_t & G^{RR}_t
\end{array}\right]$ of the matrix $G_t$. Similarly, $\tilde P_t$ is PSD because $\tilde P_t$ is the Schur complement of
$\tilde G^{LL}_t$ of the matrix $\left[\begin{array}{ll}
\tilde G^{XX}_t & \tilde G^{XL}_t \\
\tilde G^{LX}_t & \tilde G^{LL}_t 
\end{array}\right]$.

Let's now compute the value function at $t$ given by \eqref{eq:DP_V} in Theorem \ref{thm:structure}.
For notational simplicity, let 
$\alpha_t=\psi_t(\theta_t,u^R_t,\bar u^L_t,q_t, \emptyset)$.
%$X_{t+1}^{\alpha} :=X_{t+1}^{(\theta,u^R_t,\bar u^L_t,q_t)}$.

We first consider the second term of the value function in \eqref{eq:DP_V}. From the induction hypothesis we have
\begin{align}
 &(1-p)\int_{\R^{n_X}}  V_{t+1} \big(\varphi(x_{t+1})\big)\alpha_t(dx_{t+1}) \nonumber\\
=&(1-p)\int_{\R^{n_X}}  QF\left(P_{t+1},  x_{t+1} \right) \alpha_t(dx_{t+1}) 
+(1-p) e_{t+1} \nonumber\\
=&(1-p)QF\left(P_{t+1},  \mu(\alpha_t)  \right)
\nonumber\\
&+(1-p)\tr\left(P_{t+1}\cov(\alpha_t) \right)
+(1-p) e_{t+1} .
\label{eq:Vterm2}
\end{align}
The last equality in \eqref{eq:Vterm2} follows from the property of covariance.
Similarly, the last term of \eqref{eq:DP_V} becomes
\begin{align}
p\,V_{t+1}\left(\alpha_t\right) 
=&p\, QF\left(P_{t+1},  \mu(\alpha_t)  \right)
\notag \\
&+ p\tr\left(\tilde P_{t+1}\cov(\alpha_t) \right)
 +p\, e_{t+1}.
 \label{eq:Vterm3}
\end{align}
%\begin{align}
%&p\,V_{t+1}\left(\mu\left(X_{t+1}^{\alpha}\right)\right) \nonumber\\
%=&p \tr\left(P_{t+1}  \ee\left[X_{t+1}^{\alpha}(X_{t+1}^{\alpha})^\tp\right]  \right) \nonumber\\
%&+ p\tr\left((\tilde P_{t+1} - P_{t+1}) \cov(X_{t+1}^{\alpha}) \right)
% +p\, e_{t+1}\nonumber\\
%=&p \ee\left[\tr\left(P_{t+1}  X_{t+1}^{\alpha}(X_{t+1}^{\alpha})^\tp  \right)\right] \nonumber\\
%&+ p\tr\left((\tilde P_{t+1} - P_{t+1}) \cov(X_{t+1}^{\alpha}) \right)
% +p\, e_{t+1}.
% \label{eq:Vterm3}
%\end{align}
Let $S^{\theta_t}_t := \vecc(X^{\theta_t},\bar u^L_t+q_t(X_t^{\theta_t}),u^R_t)$ where 
$X^{\theta_t}$ is a random vector with distribution $\theta_t$ such that $X^{\theta_t}$ and $W_t$ are independent.
Note that from \eqref{eq:psit} in Lemma \ref{lm:beliefupdate}
\begin{align}
Y^{\theta_t}_t:=
&[A, B^L, B^R]S^{\theta_t}_t +W_t \nonumber\\
= &AX^{\theta_t} + B^L(\bar u^L_t+q_t(X^{\theta_t}))+B^R u^R_t +W_t
\end{align}
is a random vector with distribution $\alpha_t$.
Then, combining \eqref{eq:Vterm2} and \eqref{eq:Vterm3}, the last two terms of the value function becomes
\begin{align}
&QF\left(P_{t+1},  \mu(\alpha_t)  \right) \nonumber\\
&+ \tr\left(((1-p) P_{t+1}+p\tilde P_{t+1}) \cov(\alpha_t) \right)+ e_{t+1} 
\nonumber\\
=&QF\left(P_{t+1},  \ee\left[Y^{\theta_t}_t\right]  \right) \nonumber\\
&+ \tr\left(((1-p) P_{t+1}+p\tilde P_{t+1}) \cov(Y^{\theta_t}_t) \right)+ e_{t+1} 
\nonumber\\
=&QF\left(H_{t},  \ee\left[S^{\theta_t}_t\right]  \right)
+ \tr\left(((1-p) H_{t}+p\tilde H_{t}) \cov(S^{\theta_t}_t) \right)
\nonumber\\
&+ \tr\left(((1-p) P_{t+1}+p\tilde P_{t+1}) \cov(\pi_{W_t})\right)
+ e_{t+1} 
\nonumber\\
=&QF\left(H_{t},  \ee\left[S^{\theta_t}_t\right]  \right)
+ \tr\left(((1-p) H_{t}+p\tilde H_{t}) \cov(S^{\theta_t}_t) \right)
\nonumber\\
&+ e_{t}  .
 \label{eq:Vterms2and3}
\end{align}
%\begin{align}
%&\ee\left[ \tr\left(P_{t+1}  X_{t+1}^{\alpha} (X_{t+1}^{\alpha})^\tp  \right)\right] \nonumber\\
%&+ p\tr\left((\tilde P_{t+1} - P_{t+1}) \cov(X_{t+1}^{\alpha}) \right)+ e_{t+1} 
%\nonumber\\
%=&\ee\left[ \tr\left(H_t  S_t^{\alpha}  (S_t^{\alpha} )^\tp  \right)\right] 
%+ \tr\left(P_{t+1}  \Sigma_{W_t}\right)
%\nonumber\\
%&+ p\tr\left((\tilde H_{t} - H_{t}) \cov(S_t^{\alpha} ) \right)
%\nonumber\\
%&+ p\tr\left((\tilde P_{t} - P_{t})  \Sigma_{W_t}\right)
%+ e_{t+1}
%\nonumber\\
%=&\ee\left[ \tr\left(H_t  S_t^{\alpha}  (S_t^{\alpha} )^\tp  \right)\right] 
%+ p\tr\left((\tilde H_{t} - H_{t}) \cov(S_t^{\alpha} ) \right)
%\nonumber\\
%&+ e_{t+1}  .
% \label{eq:Vterms2and3}
%\end{align}

Using the random vector $S^{\theta_t}_t$, we can write the first term of the value function as
\begin{align}
&\int_{\R^{n_X}} c_t (x_t, \bar u^L_t + q_t(x_t), u_t^R) \theta_t(dx_t)
=\ee \left[ QF\left(R_t,S^{\theta_t}_t \right)\right]
\nonumber\\
=& QF\left(R_t,\ee \left[S^{\theta_t}_t \right]\right)
+ \tr\left(R_t\cov(S^{\theta_t}_t) \right)
\label{eq:Vterm1}
\end{align}

Now putting \eqref{eq:Vterms2and3} and \eqref{eq:Vterm1} into \eqref{eq:DP_V} we get
\begin{align}
V_t(\theta_t) = &e_t +\min_{q_t \in \mathcal{Q}^{\theta_t}}
\Big\{ \min_{u^R_t,\bar u^L_t} \Big\{ 
\nonumber\\
&QF\left(G_t,\ee \left[S^{\theta_t}_t \right]\right)
+ \tr\left(\tilde G_t\cov(S^{\theta_t}_t) \right)
\Big\}\Big\}.
\label{eq:DP_V_step1}
\end{align}

Note that $\ee [q_t(X^{\theta_t})] =0$ since $q_t \in \mathcal{Q}^{\theta_t}$, and consequently,
$\ee\left[S^{\theta_t}_t\right]= \vecc(\mu(\theta_t),\bar u^L_t,u^R_t)$ depends only on $u^R_t,\bar u^L_t$.
Furthermore, $\cov(S^{\theta_t}_t) = \cov\left(\vecc(X^{\theta_t},q_t(X^{\theta_t}), 0)\right)$ depends only on the choice of $q_t$. Consequently,
the optimization problem in the \eqref{eq:DP_V} can be further simplified to be
\begin{align}
V_t(\theta_t)
=  &e_{t}+
 \min_{u^R_t,\bar u^L_t} 
QF\left(G_{t},\vecc(\mu(\theta_t),\bar u^L_t,u^R_t)\right) 
\nonumber\\
+&\min_{q_t\in\mathcal{Q}^{\theta_t}}
\tr\left(\tilde G_{t}\cov\left(\vecc\left(X^{\theta_t},q_t(X^{\theta_t}),0\right)\right)  \right).
\label{eq:DP_V_step3}
\end{align}

Now we need to solve the two optimization problems
\begin{align}
 &\min_{u^R_t,\bar u^L_t} 
QF\left(G_{t},\vecc(\mu(\theta_t),\bar u^L_t,u^R_t)\right)  ,
\label{eq:DP_V_P1}
\\
&\min_{q_t\in\mathcal{Q}^{\theta_t}}
\tr\left(\tilde G_{t}\cov\left(\vecc\left(X^{\theta_t},q_t(X^{\theta_t}),0\right)\right)  \right).
\label{eq:DP_V_P2}
\end{align}
Since $G_{t}$ is PD, it follows by Lemma \ref{lm:quadratic_problems} that the optimal solution of \eqref{eq:DP_V_P1} is given by \eqref{eq:opt_ubar} and
\begin{align}
\label{eq:DP_V_part1}
\min_{u^R_t,\bar u^L_t} 
QF\left(G_{t},\vecc(\mu(\theta_t),\bar u^L_t,u^R_t)\right) 
=QF\left(P_t, \mu(\theta_t)\right).
\end{align}
Similarly, since $\tilde G_{t}$ is also PD, Lemma \ref{lm:quadratic_problems} implies that the optimal solution of \eqref{eq:DP_V_P2} is given by \eqref{eq:opt_gamma} and
\begin{align}
 &\min_{q_t\in\mathcal{Q}^{\theta_t}}
\tr\left(\tilde G_{t}\cov\left(\vecc\left(X^{\theta_t},q_t(X^{\theta_t}),0\right)\right)  \right)
\nonumber\\
=&
\tr\left(\tilde P_{t} 
\cov \left(\theta_t\right)
\right).
\label{eq:DP_V_part2}
\end{align}
Finally, substituting \eqref{eq:DP_V_part1} and \eqref{eq:DP_V_part2} into \eqref{eq:DP_V_step3} we obtain 
the \eqref{eq:Vt_PacketDrop} at $t$.
%\begin{align}
%V_t(\theta_t)
%= &e_{t} + \tr\left(P_t \ee[X_t^{\theta_t} (X_t^{\theta_t})^\tp]\right) - \tr\left(P_t\cov\left( X_t^{\theta_t}\right)\right)
%\nonumber\\
%&+\tr\left(\tilde P_t\cov\left(X_t^{\theta_t}\right)\right).
%\end{align}
This completes the proof of the induction step and the proof of the theorem.
\end{proof}

\begin{proof}[Proof of Theorem \ref{thm:opt_strategies}]
Let $\hat X_t$ be the estimate (conditional expectation) of $X_t$ based on the common information $H^R_t$. Then, for any realization $h^R_t$ of $H^R_t$, $\hat x_t = \mu(\theta_t)$. This together with Theorems \ref{thm:structure} and \ref{thm:Sol_packetdrop} result in \eqref{eq:opt_ULbarUR} and \eqref{eq:opt_UL}.
To show \eqref{eq:estimator_0} and \eqref{eq:estimator_t}, note that 
%Then, according to Theorem \ref{thm:Sol_packetdrop}, we only need to show that $\hat x_t$ satisfy 
%\eqref{eq:estimator_0} and \eqref{eq:estimator_t}. 
at time $t=0$, for any realization $h^R_t$ of $H^R_t$,
\begin{align}
&\hat x_0 = \mu(\theta_0) = \int_{\R^{n_X}} y \theta_0(dy) \notag \\
&= 
\Big \lbrace \begin{array}{ll}
\int_{\R^{n_X}} y \pi_{X_0}(dy) = \mu(\pi_{X_0}) & \text{ if }z_{0}= \emptyset,\\
 \int_{\R^{n_X}} y \varphi (x_0)(dy) =  
 %x_0 \int_{\R^{n_X}} \varphi (x_0)(dy)  = 
 x_0 & \text{ if }z_{0} = x_0.
\end{array} 
\end{align}
Therefore, \eqref{eq:estimator_0} is true. Furthermore, at time $t+1$ and for any realization $h^R_t$ of $H^R_t$,
\begin{align*}
\hat x_{t+1} = \mu(\theta_{t+1}) = \int_{\R^{n_X}} y \psi_t(\theta_t,u^R_t,\bar u^L_t,q_t, z_{t+1}) (dy).
\end{align*}
If $z_{t+1} = x_{t+1}$, then $\hat x_{t+1} = \int_{\R^{n_X}} y \varphi(x_{t+1}) (dy) = x_{t+1}.$
%\begin{align}
%\hat x_{t+1} = \int_{\R^{n_X}} y \varphi(x_{t+1}) (dy) = x_{t+1}.
%\end{align}
\\
If $z_{t+1} = \emptyset$, then,
\begin{align}
&\hat x_{t+1} = \int_{\R^{n_X}}\! y \psi_t(\theta_t,u^R_t,\bar u^L_t,q_t, \emptyset) (dy)
= \int_{\R^{n_X}}\! y\! \int_{\R^{n_X}}\! \int_{\R^{n_X}} 
\notag \\
&
\mathds{1}_{\{y\}}(Ax_t + B^L (\bar u_t^L + q_t (x_t)) +  B^R u^R_t + w_t) 
\notag \\
&
\theta_t(dx_t) \pi_{W_t}(dw_t) dy
\notag \\
& = 
\int_{\R^{n_X}} \int_{\R^{n_X}} (Ax_t + B^L (\bar u_t^L + q_t (x_t)) +  B^R u^R_t + w_t) 
\notag \\
&
\theta_t(dx_t) \pi_{W_t}(dw_t)
%= \int_{\R^{n_X}} A x_t \theta_t(dx_t)  + B^L \bar u_t^L 
%\notag \\
%&+ \int_{\R^{n_X}} B^L q_t (x_t) \theta_t(dx_t) +  B^R u^R_t +
%\int_{\R^{n_X}} w_t \pi_{W_t}(dw_t)
= A\hat x_t + B^L \bar u_t^L + B^R u^R_t.
\end{align}
where the third equality is true because 
\begin{align*}
&\int_{R^{n_X}} y \mathds{1}_{\{y\}}(Ax_t + B^L (\bar u_t^L + q_t (x_t)) +  B^R u^R_t + w_t)dy
\notag \\
&= Ax_t + B^L (\bar u_t^L + q_t (x_t)) +  B^R u^R_t + w_t.
\end{align*}
Furthermore, the last equality is true because $q_t \in \mathcal{Q}^{\theta}$ and $W_t$ is a zero mean random vector.
Therefore, \eqref{eq:estimator_t} is true and the proof is complete.
\end{proof}

\end{document}